\documentclass[runningheads]{llncs}

\usepackage[T1]{fontenc}
\usepackage{csquotes}
\usepackage[english]{babel}
\usepackage{indentfirst}

\usepackage{bm}
\usepackage{amsfonts}
\usepackage{graphicx}
\usepackage[colorlinks=true, allcolors=blue]{hyperref}
\usepackage{xcolor}
\usepackage{comment}
\usepackage{mathtools}
\usepackage[ruled, linesnumbered, noend]{algorithm2e}
\usepackage{array}

\usepackage{thmtools}
\usepackage{thm-restate}

\usepackage[colorinlistoftodos,prependcaption,textsize=tiny]{todonotes}
\usepackage{xargs}

\newcommand{\CFM}{\approx_{CF}}

\newcommand{\PDD}{\protect\operatorname{PD}}
\newcommand{\FPDD}{\protect\operatorname{F-PD}}

\newcommand{\sx}{\operatorname{small}_x}
\newcommand{\ex}{\operatorname{equal}_x}

\newcommand{\PD}{Parent-Distance}
\newcommand{\FPD}{Forest Parent-Distance}
\newcommand{\SN}{SN}
\newcommand{\FSN}{\protect\operatorname{F-SN}}
\newcommand{\Skipped}{Skipped-Number}
\newcommand{\FSkipped}{Fo\-rest Skip\-ped-Num\-ber}
\newcommand{\Fil}{\mathcal{F}il}
\newcommand{\SNRepres}{{\em Skipped-number} representation}
\newcommand{\REFD}[1]{\protect\operatorname{F-ref}_{#1}} 

\newcommand{\CT}{Car\-te\-sian Tree}
\newcommand{\CF}{Cartesian Forest}
\newcommand{\ST}{Schröder Tree}

\newcommand{\rb}{\operatorname{rb}}  
\newcommand{\rbs}{\operatorname{rbs}}  

\newcommand{\RMP}{\operatorname{RB}} 

\newcommand{\leftsb}{\operatorname{left}}   
\newcommand{\rightsb}{\operatorname{right}} 

\newcommand{\algoName}[1]{\textsc{#1}}
\newcommand{\LN}{\protect\operatorname{LN}} 
\newcolumntype{P}[1]{>{\centering\arraybackslash}p{#1}}

\newcommandx{\RQunsure}[2][1=]{\todo[linecolor=red,backgroundcolor=red!25,bordercolor=red,#1]{#2}}
\newcommandx{\RQchange}[2][1=]{\todo[linecolor=blue,backgroundcolor=blue!25,bordercolor=blue,#1]{#2}}
\newcommandx{\RQinfo}[2][1=]{\todo[linecolor=OliveGreen,backgroundcolor=OliveGreen!25,bordercolor=OliveGreen,#1]{#2}}
\newcommandx{\RQimprovement}[2][1=]{\todo[linecolor=Plum,backgroundcolor=Plum!25,bordercolor=Plum,#1]{#2}}

\def\ctbord{CTBord}%

\begin{document}
\title{\CF~Matching}
%
%
\author{Bastien Auvray\inst{1,2} \and
Julien David\inst{1,3} \and
Richard Groult\inst{1,2} \and
Thierry Lecroq\inst{1,2}} 
\authorrunning{B. Auvray et al.}
%
\institute{CNRS NormaSTIC FR 3638, France \and
Univ Rouen Normandie, LITIS UR 4108, F-76000 Rouen, France
\and Normandie University, UNICAEN, ENSICAEN, CNRS, GREYC,
Caen, France}
\maketitle              
%

\begin{abstract}
In this paper, we introduce the notion of \CF, which generalizes
\CT s, in order to deal with partially ordered sequences.
We show that algorithms that solve both exact and approximate \CT~Matching
can be adapted to solve \CF~Matching in linear time in the worst-case for exact matching and linear time on average for approximate matching.
We adapt the notion of \CT~Signature to \CF s.
We also show a one-to-one correspondence between \CF s and \ST s.
\keywords{\CT ~\and \CF ~\and Pattern Matching \and Approximate Pattern Matching \and \ST}
\end{abstract}

\section{Introduction}
Pattern matching consists of searching for one or all the 
 occurrences of a pattern in a text.
It is an essential task in many computer science applications.
It can take different forms.
For instance it can be done online when the pattern can be
 preprocessed or offline when the text can be preprocessed.
Occurrences can be exact or approximate.
When the pattern and the text are sequences of characters,
 it is known as string matching.
When searching patterns in time series data, the notion
 of pattern matching is a bit more involved.
Solutions can use \CT s that where introduced
 by Vuillemin in 1980~\cite{vuillemin80}.
 
\CT ~Matching has been introduced by Park \textit{et al.}~\cite{PALP19,PARK2020}.
Given a pattern, it consists of finding the factors of a text
 that share the same \CT ~as the \CT ~of the pattern.

Since then it has gained a lot of interest.
Efficient solutions for practical cases for online search
 were given in~\cite{SGRFLP21}.
Expected linear time algorithms are given
 in~\cite{auvray2025approximatecartesiantreematching} for approximate
 \CT ~Matching with one difference.

Indexing structures in the \CT ~pattern matching 
 framework are presented in~\cite{NFNI2021,KimC21,osterkamp_et_al:LIPIcs.CPM.2025.26}.
Methods for computing regularities are given in~\cite{KikuchiHYS20}
 and methods for computing palindromic structures are presented
 in~\cite{Funakoshi_et_al}.
An algorithm for episode matching
 (given two sequences $p$ and $t$, finding all minimal length factors of $t$ that contains $p$ as a subsequence) in \CT ~framework is presented in~\cite{OizumiKMIA22}.
Practical methods for finding longest common Cartesian substrings of two strings appeared in~\cite{FLPS22}.
Very recently, dynamic programming approaches for approximate \CT ~pattern matching with edit distance have been
 considered in~\cite{KimH24} and longest common \CT
~subsequences are computed in~\cite{TsujimotoSMNI24}.
 Efficient algorithms for determining if two equal-length indeterminate
 strings match in the \CT ~framework are given in~\cite{gawrychowski_et_al:LIPIcs.CPM.2020.14}.

\CT s are defined (Section~\ref{sec:ct}) on arbitrary sequences where there can be values with multiple occurrences.
In that case, ties should
be broken in some way.
In this article, in order to better deal with equal values we introduce \emph{\CF s} (Section~\ref{sec:cf}).

We show that algorithms that solve both exact and approximate \CT~Matching (Section~\ref{sec:CFM} and Section~\ref{sec:ACFM})
 can be adapted to solve \CF~Matching in linear time in the worst-case for exact matching and linear time on average for approximate matching.
We then show a one to one correspondence between \CF s and \ST s (Section~\ref{sec:combi}).
More specifically we give algorithms for computing a \ST ~and
 a special type of words with parentheses given a \CF.
We also adapt the notion of \CT~Signature introduced in~\cite{demaine09}
 to \CF s (Section~\ref{sec:sign}), and show
 how this notion can be used to experimentally improve the computation (Section~\ref{sec:XP}).

\section{Preliminaries \label{sec:ct}}

In this paper, a sequence is always defined
on an ordered alphabet.
For a given sequence $x$, $\vert x\vert$ denotes the length of $x$.
A sequence $v$ is \emph{factor} of a sequence $x$ if $x=uvw$ for any sequences $u$ and $w$.
A sequence $u$ is a \emph{prefix} (resp. \emph{suffix}) of a sequence $x$ if $x=uv$ (resp. $x=vu$).
A prefix of a sequence that is also a suffix
 is called a \emph{border}.
For a sequence $x$ of length $m$,
 $x[i]$ is the $i$-th element of $x$ and $x[i\ldots j]$ represents the factor of $x$ starting at the $i$-th element and ending at the $j$-th element, for $1\le i\le j\le m$. 

\subsection{\CT s}

\begin{definition}[\CT ~$C(x)$] \label{def:ct}
    Given a totally ordered sequence $x$ of length $m$, the \emph{\CT} of $x$, denoted by $C(x)$, is the binary tree recursively defined as follows:
    \begin{itemize}
        \item if $x$ is empty, then $C(x)$ is the empty tree;
        \item if $x[1\ldots m]$ is not empty and $x[i]$ is the smallest value of $x$, $C(x)$ is the binary tree with $i$ as its root, the \CT ~of $x[1\ldots i-1]$ as the left subtree and the \CT ~of $x[i+1\ldots m]$ as the right subtree.
    \end{itemize}
\end{definition}

We denote by $\rb(T)$ the list of nodes on the \emph{right branch} of a \CT ~$T$.
Also, $C_i(x) = C(x[1\ldots i])$ to simplify the notations.
The \CT ~of a sequence can be built online in linear time and space~\cite{CTtime}.
Informally,
 let $x$ be a sequence such that $C_{i-1}(x)$ is already known. 
In order to build $C_i(x)$, one only needs to find the nodes $j_1<\cdots<j_k$ in $\rb(C_{i-1}(x))$ such that $x[j_1] > x[i]$.
Then, $j_1$ will be the root of the left subtree of node $i$.
If $j_1$ is the root of $C_{i-1}(x)$ then $i$ becomes the root of $C_i(x)$,
 otherwise let $j_0$ be the parent node of $j_1$, then node $i$ will be the root of the right subtree of $j_0$.
Then $\rb(C_i(x))=\rb(C_{i-1}(x))\setminus(j_1,\ldots,j_k) \cup (i)$.
All these operations can be easily done by implementing $\rb$ with a stack.
The amortized cost
 of such an operation can be shown to be constant.

\CT s are defined on arbitrary sequences where there can be multiple occurrences
 of the smallest value. In that case, ties should
 be broken in some way.
Usually the first occurrence of the smallest value is chosen to be the root of the tree.

Given a \CT\ $C(x)$ of a sequence $x$ of length $m$, 
a linear representation of $C(x)$,
denoted $\LN_x$, is a table of integers of length $m$, and 
the set of linear representations of a given length
is in bijection with the set of \CT s of the same size.
A linear representation of \CT s was introduced by Park \textit{et al.}~\cite{PALP19}: the parent-distance representation (see example \figurename~\ref{fig:parent}).
Using this representation and classical string algorithms (Knuth-Morris-Pratt~(KMP)~\cite{KMP77} and Aho-Corasick~\cite{AC75} algorithms), they obtained linear-time solutions for single and multiple pattern \CT\ matching (see \figurename~\ref{fig:ln_bord} for an example).

Let $\sx(i) = \max(\{k \mid x[k] < x[i] \mbox{ with } 1\leq k<i \}\cup\{0\})$.

\begin{definition}[Parent-distance representation $\PDD_x$] \label{def:pd}
Given a sequence $x$ of length $m$, the \emph{parent-distance representation} of $x$ is an integer sequence $\PDD_x[1\ldots m]$, which is defined as follows:

$
\PDD_x[i] = 
\begin{cases}
i - \sx(i) & \text{if } \sx(i)>0 \cr
0 & \text{otherwise.}
\end{cases}
$
\end{definition}

Since the parent-distance representation has a one-to-one mapping with \CT s, it can replace them without loss of information.

The parent-distance table is not the only linear representation of \CT s.
Another linear representation is based on the number of nodes skipped by a new node when building online the \CT\ of a sequence.
Informally, for each node $i$, it stores the number of nodes $(j_1 < \cdots < j_k)$ deleted from the right path when computing $C_i(x)$ from $C_{i-1}(x)$.
We say that node $i$ skipped nodes $(j_1 < \cdots < j_k)$ and that these nodes are skipped by node $i$.  
The left subtree of a \CT\ $C(x)$ is denoted
$left(C(x))$ and
the length of its right branch $\RMP(C(x)) = |\rb(C(x))|$. 

See \figurename~\ref{fig:ln_bord} for an example.

\begin{definition}[Skipped-number representation $\SN_x$] \label{def:SN}
Given a sequence $x[1\ldots m]$, the {\em Skipped-number representation} of $x$ is a se\-quen\-ce $\SN_x[1\ldots m]$
of integers
such that $\SN_x[i]$ is the length of the right-branch of the left-subtree of the subtree rooted at node $i$ of the \CT\ of $x$, that is
$\SN_x[i] = \RMP( \leftsb(C_i(x)) )$.
\end{definition}
Both linear representations can be computed in linear time~\cite{PALP19,demaine09,auvray2025approximatecartesiantreematching}.

\subsection{Exact \CT\ Matching}

Given two sequences $x$ and $y$ of length  $m$ and $n$ respectively, the \emph{Exact \CT\ Matching} problem (CTM) consists
 of finding all the factors of $y$ that have
 the same \CT\ of the one of $x$.
This can be done in time $O(m+n)$ by
 using a KMP-like algorithm.
To achieve this complexity, a \emph{Cartesian border table} is used.
Given the Parent-Distance representation of $x$,
 the Parent-Distance of a factor $x[i\ldots j]$ of $x$ satisfies:

$\PDD_{x[i\ldots j]}[k] =
\begin{cases}
0 & \mbox{if } \PDD_x[i+k-1] \geq k,\cr
\PDD_x[i+k-1] & \mbox{otherwise}.
\end{cases}$

\begin{definition}[Cartesian Border Table $\ctbord$]
The Cartesian border table of $x$ is defined in the following way:
$\ctbord[0] = 0$ and, for $1 \leq k < i$, 
$\ctbord[i] = \max\{k \mid C_k(x) = C(x[i-k+1\ldots i])\}$.
\end{definition}

\figurename~\ref{fig:ln_bord} shows the Cartesian border table of
$x = (3, 1, 6, 4, 8, 6, 7, 5, 9$).

\begin{figure}[!t]
\begin{center}
\begin{tabular}{|l||c|c|c|c|c|c|c|c|c|}
\hline
$i$ & 1 & 2 & 3 & 4 & 5 & 6 & 7 & 8 & 9\\
\hline \hline
$x[i]$ & 3 & 1 & 6 & 4 & 8 & 6 & 7 & 5 & 9\\
\hline
$\PDD_x[i]$ & 0 & 0 & 1 & 2 & 1 & 2 & 1 & 4 & 1\\
\hline
$\SN_x[i]$ & 0 & 1 & 0 & 1 & 0 & 1 & 0 & 2 & 0\\
\hline
$\ctbord[i]$ & 0 & 1 & 1 & 2 & 3 & 4 & 5 & 2 & 3\\
\hline
\end{tabular}
\caption{\label{fig:ln_bord} The parent-distance representation, the skipped-number representation and the Cartesian border table of a sequence $x$.}
\end{center}
\end{figure}

\subsection{Approximate \CT\ Matching}

Given two sequences $x$ and $y$ of length  $m$ and $n$ respectively, the \emph{Approximate \CT\ Matching} problem (ACTM) consists
 of finding all the factors of $y$ that have
 the same \CT\ of the one of $x$ up to some differences.

Unfortunately, the linear-time solutions obtained in~\cite{PALP19} for the exact case do not work anymore in the general approximate case. Also we conjecture that there is no
algorithm that solves approximate pattern matching with a linear
worst-case complexity.

In~\cite{auvray2025approximatecartesiantreematching}, the authors provide the following MetaAlgorithm that can solve ACTM
with up to one difference (substitution, insertion, deletion, swap) in quadratic
worst-case time complexity and in linear time in average.
The test at line $5$, using the $\sim$ symbol, signifies that depending on the considered  notion of error, we will not test equality between linear representation
but rather test an equivalence.

\begin{algorithm}[h!]\label{algo:meta}
  \caption{\algoName{metaAlgorithm}($p$, $t$)}
  \SetAlgoLined
  \SetKwInOut{KwIn}{Input}
  \SetKwInOut{KwOut}{Output}
  \KwIn{Two sequences $p$ and $t$ of length $m$ and $n$}
  \KwOut{The number of positions $j$ such that $p$ is equivalent to $t[j\ldots j+m-1]$ }
  $occ \leftarrow 0$\\
  $x \leftarrow t[1\ldots m]$\\
  $\LN_{p}, \LN_{x} \leftarrow$ linear representations of $C(p)$ and $C(x)$\\
  \For{$j \in \{1,\ldots,n-m+1\}$}{
    \If{$\LN_{p} \sim \LN_{x}$}{
        $occ \leftarrow occ +1$
    }
    $x \leftarrow t[j+1\ldots j+m]$\\
    Update $\LN_{x}$\\
  }
  \Return{$occ$}
\end{algorithm}

\section{Cartesian Forests \label{sec:cf}}
A limitation of \CT s is that it does not take into
account equal values in the original sequence.
For instance, the sequence $(1,1,1,1,1)$ shares the
same \CT\ as $(1,2,3,4,5)$ when ties are broken by ordering equal values according to their positions. 
In order to better deal with equal values we now introduce \CF s.

Recall that a planar tree, or general tree, is a tree in which
each node possesses a sequence of children, or subtrees.
In the following, \emph{trees} designates planar trees.

Informally, during the construction of $C_i(x)$ from $C_{i-1}(x)$,
 assume that node~$i$ should be inserted in the right branch between
 nodes $j_0$ and $j_1$ such that
 $x[j_0] \le x[i] < x[j_1]$.
In a \CT, node~$i$ becomes the root of the right subtree of node $j_0$ and node $j_1$ becomes the root of the left subtree of node $i$.
In a \CF, node $i$ becomes the root of the right subtree of node $j_0$ only if $x[j_0] < x[i]$ while it becomes the right sibling of node
 $j_0$ if $x[j_0] = x[i]$ and node $j_1$ remains in the right subtree of node $j_0$.

A \emph{forest} is a sequence of trees. 
An empty forest is an empty sequence that does not contain any tree.
The children of a node can be seen as a \emph{sub-forest}.
In the following, a node of a tree in a \emph{\CF}, 
instead of having a sequence of children, can possess a 
{\em left} sequence  of children and a {\em right} sequence of children, which
we call the {\em left sub-forest} and the {\em right sub-forest}.
The leftmost tree of a \CF\ is actually the only one 
that can possess a left sub-forest and a right sub-forest,
whereas all the other trees can only have a right sub-forest.
\figurename~\ref{fig:forest} illustrates this with an example of a \CF.


\begin{definition}[\CF ~$F(x)$ of a sequence $x$] \label{def:cf}
    Given a sequence $x$ of length $m$, the \emph{\CF} of $x$, denoted by $F(x)$, is recursively defined as follows:
    \begin{itemize}
        \item if $x$ is empty, then $F(x)$ is the empty forest;
        \item if $x[1\ldots m]$ is not empty, let $(r_1, \ldots, r_k)$ be the ordered
        sequence of all the $k$ positions of the smallest value of $x$ and let $r_{k+1} = m+1$, 
        then $F(x)$ is a forest composed of $k$ trees whose roots are $r_1, \ldots, r_k$, such that:
        \begin{itemize}
            \item the \emph{left} sub-forest of $r_1$ is the \CF ~$F(x[1\ldots r_1-1])$, and
            \item for $1 \le i \leq k$, the \emph{right} sub-forest of $r_i$ is the \CF ~$F(x[r_i+1\ldots r_{i+1}-1])$.
        \end{itemize} 
    \end{itemize}
\end{definition}

\begin{figure}[!ht]
\begin{center}
\includegraphics[scale=0.7]{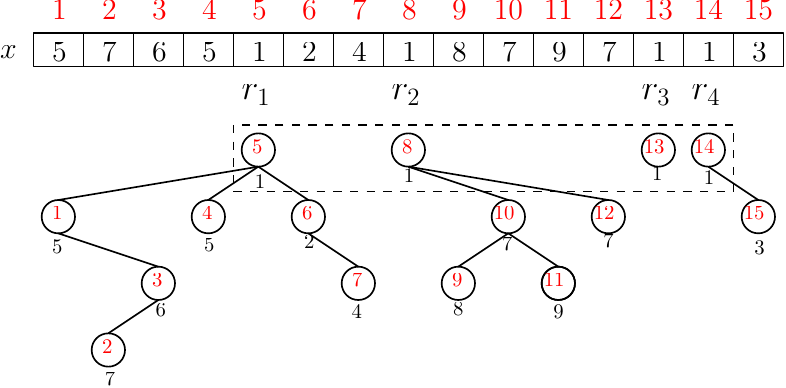}

\caption{
The \CF\ associated to an ordered sequence $x$.
\label{fig:forest}
Roots $r_2$ to $r_4$ are created since the sequence contains
a value equal to the one at position $5$.
The values between two roots $r_i$ and $r_{i+1}$ belong to the tree
enrooted in $r_i$.
Therefore, except for the leftmost tree of a \CF, 
a tree cannot have a left sub-forest. 
This idea is true at every level and can be seen in this example at positions $4$ and $12$.
}
\end{center}
\end{figure}

\section{Exact \CF~Matching \label{sec:CFM}}

We adapt the linear representations of \CT s, such as the 
\PD~\cite{PALP19} and the \Skipped ~representation~\cite{demaine09}
to \CF s.
Basically, for a sequence $x$, when computing $C_i(x)$
 from $C_{i-1}(x)$, the \PD ~of $i$ is equal to the distance between $i$ and its parent in $C_i(x)$
 and its Skipped-Number is the number of nodes removed from the right branch
 of $C_{i-1}(x)$ comparing to the right branch of
 $C_i(x)$ (see~\cite{auvray2025approximatecartesiantreematching}).

The main idea of the adaptation to \CF s is simple: both representations involve
the comparisons of two values $x[j]$ and $x[k]$ in the original sequence $x$, 
for any $j < k$.
Suppose we have either $x[j]> x[k]$ or $x[j]< x[k]$, then
 we are in the case of \CT s and the linear representation
 is the same. The idea for the \PD ~of the \CF\ of a sequence $x$ is that
 its absolute value at a position $i$ gives
the distance between a node $i$ and its parent 
or left sibling
in the \CF\ associated to the sequence $x[1\ldots i]$.
Let
$\ex(i) = \max(\{k \mid x[k] = x[i]\mbox{ with } 1\leq k<i\}\cup\{0\})$.

\begin{definition}[\FPD\ representation $\FPDD_x$] \label{def:FPD}
Given a sequence $x$ of length $m$, the \emph{\FPD\ representation} of $x$ is an integer sequence $\FPDD_x[1\ldots m]$ defined as follows:

$\FPDD_x[i] = 
\begin{cases}
i - \sx(i) &  \text{if } \sx(i) > \ex(i) \cr
-(i - \ex(i)) &  \text{if } \sx(i) < \ex(i) \cr
0 &  \text{otherwise.}
\end{cases}
$
\end{definition}

The referent table that we define next contains the position
of the next smaller (or equal) value in the sequence.
\begin{definition}[The referent table $\REFD{x}$] \label{def:skippednodes}
Given a sequence $x[1\ldots m]$, the \emph{referent table} of $x$ is a sequence of sets $\REFD{x}[1\ldots m]$
such that  

$\REFD{x}[i] = 
\begin{cases}
\min_{i< j\leq m}\{j \mid x[j] \le x[i]\} & \text{if such } j \text{ exists}\cr
-1 & \text{ otherwise.}
\end{cases}$
\end{definition}

\begin{definition}[\FSkipped ~representation $\FSN_x$] \label{def:FSN}
Given a sequence $x[1\ldots m]$, the \emph{\FSkipped ~representation} of $x$ is an integer sequence $\FSN_x[1\ldots m]$
such that 

$\FSN_x[i] = 
\begin{cases}
|\{ j < i \mid \REFD{x}[j] = i\} | & \mbox{if } \sx(i) > \ex(i) \cr
-|\{ j < i \mid \REFD{x}[j] = i\} | & \mbox{if } \sx(i) < \ex(i) \cr
\end{cases}
$
\end{definition}

\figurename~\ref{fig:parent} shows examples of \FPD\ representations, \FSkipped\  representations  and referent tables of two \CF s.

\begin{figure}[!ht]
\begin{minipage}{0.49\textwidth}
\includegraphics[scale=0.75]{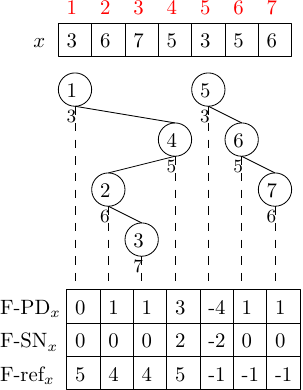}
\end{minipage}
\begin{minipage}{0.49\textwidth}
\includegraphics[scale=0.75]{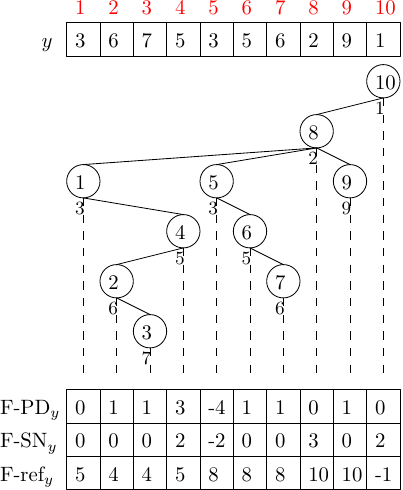}
\end{minipage}
\begin{center}

\caption{\label{fig:parent}%
Two sequences $x$ and $y$,  
their associated \CF s $F(x)$ and $F(y)$,
and their corresponding \FPD~representations, \FSkipped ~representations and referent tables. 
As one can see, $x$ is a prefix of $y$ and the forest $F(x)$ is transformed into a sequence
of left subtrees in $F(y)$.
}
\end{center}
\end{figure}
Given two sequences $x$ and $y$, we will denote $x \CFM y$ when the two sequences share the same \CF. 
\begin{definition}[\CF ~Matching (CFM)] \label{def:cfm}
Given two sequences $p[1\ldots m]$ and $t[1\ldots n]$, 
find every position $j$, with $1 \le j \le n - m + 1$, such that $t[j\ldots j+m-1] \CFM p[1\ldots m]$.
\end{definition}

\begin{example}
Let $t=(5,7,3,6,3,7,2,8,2,4,3,3)$ and $p=(2,3,1,4,1,5)$ respectively be the text and the pattern. We have two occurrences of $p$ in $t$ as $t[1\ldots 6] = (5,7,3,6,3,7)\CFM p$ and  
$t[5 \ldots 10] = (3,7,2,8,2,4) \CFM p.$
\end{example}

\begin{proposition}
Algorithm~\ref{algo:meta} solves the CFM problem in $\mathcal{O}(n)$ space and $\mathcal{O}(mn)$ time in the worst-case and $\mathcal{O}(n)$ time on average.
\end{proposition}

\begin{proposition}
Using a \CF\ border table, the CFM problem can be solved in $\mathcal{O}(n)$ time and space by adapting the KMP approach used in \cite{PALP19}.
\end{proposition}

Note that, in order to prove both propositions, it is sufficient to prove that two sequences share the same Cartesian Forest if and only if they share the same linear representation.
It can done by adapting the proof
from~\cite{PALP19}\footnote{see Appendix A.1}. 

\section{Approximate \CF ~Matching \label{sec:ACFM}}

The above approach for exact CFM can be extended to \emph{approximate CFM}
 with one difference
 in a similar way as approximate CTM with one difference is done
 in~\cite{auvray2025approximatecartesiantreematching}.
 Given a sequence $x$ of length $m$ and $i$ a position $1\le i < m$, 
let $y$ be the sequence defined by
 swapping the elements in positions $i$ and $i+1$ in $x$.
In~\cite{auvray2025approximatecartesiantreematching}, it is shown
 that there are at most $3$ mismatches between the \SNRepres~of the \CT
 ~of $x$ and the \SNRepres~of the \CT ~of $y$.
The same result holds for the \FSkipped ~representation of the \CF
 ~of $x$ and the \FSkipped ~representation of the \CF ~of $y$ considering that two
  elements match if their absolute values are equal.
Thus, based on the results
 in~\cite{auvray2025approximatecartesiantreematching},
 the \emph{CFM with one swap} of a pattern of length $m$ in a text of
 length $n$ can be done in time $O(mn)$ in the worst case and
 in linear time in average.

\emph{Approximate CTM with one mismatch, one insertion or one deletion} is solved
 by comparing 
 the \PD ~representations from left to right and 
 the \PD ~representations from right to left 
 of the \CT s of $x$ and $y$.
The \FPD ~representation from right to left of a \CF ~can be
 defined in a similar way as for a \CT.
And then, \emph{approximate CFM with one mismatch, one insertion or one deletion}
 can be solved in a similar way than approximate CTM with one mismatch, one insertion or one deletion.
Thus, based on the results
 in~\cite{auvray2025approximatecartesiantreematching},
 the CFM with one mismatch, one insertion or one deletion
 of a pattern of length $m$ in a text of
 length $n$ can be done in time $O(mn)$ in the worst case and
 in linear time in average.

\section{Combinatorics of the \CF s \label{sec:combi}} 
\subsection{Recursive Definition}
\begin{definition}[\CF ~$F$] \label{def:purecf}
A \CF\ $F$ can be:
\begin{itemize}
    \item empty,
    \item a sequence of $k$ trees rooted in $(r_1,\ldots,r_k)$, with $k\ge 1$, such that 
    $r_1$ has a left Cartesian sub-Forest and a right Cartesian sub-Forest, and
    $r_i$ has only a right Cartesian sub-Forest for all $i \in\{2,\ldots, k\}$ (the left sub-Forest is necessarily empty).
\end{itemize}
\end{definition}

\figurename~\ref{fig:valid_forest} shows examples of one \CF ~and two forests that are not Cartesian.

\begin{figure}[!ht]
\begin{center}
\includegraphics[scale=0.6]{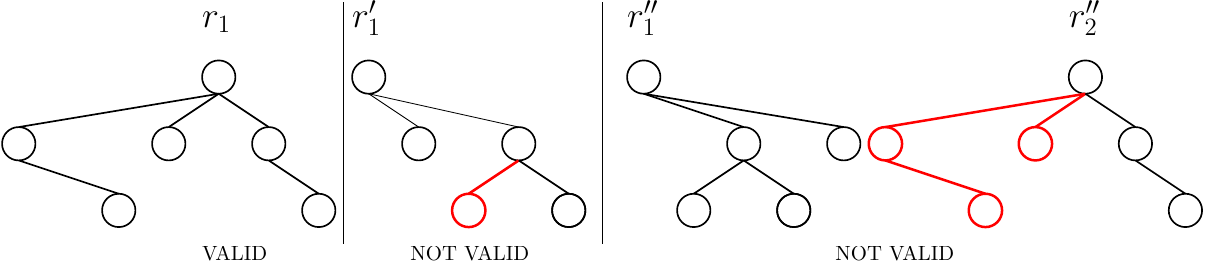}
\caption{
On the left is a valid \CF. 
In the middle, it is not a valid \CF ~because the second tree in the right sub-forest of $r'_1$ has a left sub-forest.
On the right, the forest is not Cartesian: the second tree has a left sub-forest.
\label{fig:valid_forest}
}
\end{center}
\end{figure}

According to Definition~\ref{def:purecf}, a \CF ~$F$ can be either empty (denoted by $\emptyset$) or contains
at least one node. This node - the leftmost root - has a left sub-forest 
and a right sub-forest, but it can also have siblings.
A sibling $S$ is a particular \CF ~that cannot
have a left sub-forest. Therefore, we obtain
the following recursive decomposition:
$
\begin{cases}
\includegraphics[scale=0.2]{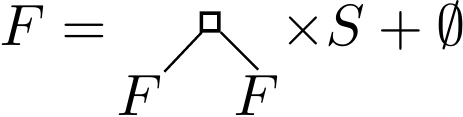}\\
\includegraphics[scale=0.2]{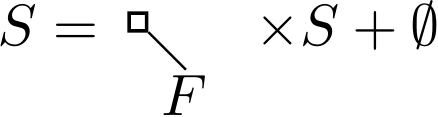}\\
\end{cases}
$

\subsection{Generating Function}

Let $F(z) = \sum_{n\ge 0} f_n z^n$ be the generating function of \CF s, where $f_n$ counts the number of \CF s with $n$ nodes, and
$S(z)$ be the generating function of siblings.
The translation between the recursive decomposition is made using
the Symbolic Method~\cite{flajolet2009analytic}:
$$
\begin{cases}
F(z) = z\cdot F^2(z)\cdot S(z) +1 \\
S(z) = z\cdot F(z) \cdot S(z) +1
\end{cases}
$$ 
from which we obtain that 
$S(z) = \frac{1}{1-z\cdot F(z)}$ and  $F(z)  = 1 + \frac{z\cdot F^2(z)}{1-z\cdot F(z)}$.

It can be shown from here that $F(z) = \frac{1+z+\sqrt{1-6z+z^2}}{4z}$, whose 
associated coefficients are known to be Schröder–Hipparchus numbers, also called super-Catalan numbers (Sequence \href{https://oeis.org/A001003}{$A001003$} on OEIS) enumerated by the following formula:
$f_n= \sum_{i=1}^n \frac1n \binom{n}{i}\binom{n}{i - 1}2^{i - 1}$,
which, amongst other things, counts the number of ways of 
inserting parentheses into a sequence of $n+1$ symbols, where each pair of parentheses surrounds at least two symbols or parenthesized groups, and without any parentheses surrounding the entire sequence.
In~\cite{bouvel2011average}, the authors show that: $$f_n \sim \frac{\sqrt{3\sqrt{2}-4}}{4\sqrt{n^3\pi}}(3+2\sqrt{2})^n \sim 0.07 (5.828)^n n^{-\frac32}.$$

In the following subsections, we describe bijections between
\CF s and classical combinatorial objects counted by the Schröder-Hipparchus numbers. We recall that since those objects all share the same generating function, there exists a bijection between those sets\footnote{\figurename~\ref{fig:bijection} (Appendix)
shows the one to one correspondence for $n \in \{1,2,3\}$ between \CT s, \ST s and Parentheses Words}. 
Any injection from one set to the other is a bijection.

\subsection{A bijection with Schröder Trees}

\begin{definition}[\ST~$ST$] \label{def:st}
A \ST ~is a tree whose internal nodes have two or more subtrees.
\end{definition}
\ST s with $n+1$ leaves are counted by $f_n$. 
Given a \CF ~$F$ with $k$ roots $(r_1, \ldots, r_k)$,
we denote $\leftsb(r_1)$ the left sub-forest of the root $r_1$.
For $i \in \{1,\ldots, k\}$, we denote $\rightsb(r_i)$ the right sub-forest of $r_i$.
From a \CF ~with $n\geq 1$ nodes, Algorithm~\ref{algo:CF_to_ST} gives a recursive approach to building a \ST ~with $n+1$ leaves.

\begin{figure}[ht!]
\begin{minipage}{0.49\textwidth}
\begin{algorithm}[H]\label{algo:CF_to_ST}
  \caption{\algoName{CFtoST}$(F)$}
  \SetAlgoLined
  \SetKwInOut{KwIn}{Input} 
  \SetKwInOut{KwOut}{Output} 
  \KwIn{A \CF~$F$ w. $n$ nodes and $k$ roots, s.t. $F = (r_1,\ldots,r_k)$}
  \KwOut{A \ST with $n+1$ leaves}
  $ST \leftarrow$ a root\;
  \If{$F$ is not empty}{
  $(c_1, \ldots, c_{k+1}) \leftarrow$ Create a tuple of subtrees\;
  $c_1 \leftarrow$ \algoName{CFtoST}$(\leftsb(r_1))$\;
  \For{$i \in \{2, \ldots, k+1\}$}{
  $c_i \leftarrow$ \algoName{CFtoST}$(\rightsb(r_{i-1}))$\;
  }
  $ST \leftarrow$ Add subtrees $(c_1, \ldots, c_{k+1})$ to $ST$\;
  }
  \Return{$ST$}\;
\end{algorithm}
\end{minipage}
\begin{minipage}{0.49\textwidth}
\begin{algorithm}[H]\label{algo:CF_to_W}
  \caption{\algoName{CFtoW}($F$)}
  \SetAlgoLined
  \SetKwInOut{KwIn}{Input} 
  \SetKwInOut{KwOut}{Output} 
  \KwIn{A \CF~$F$ with $n$ nodes and $k$ roots, s.t. $F = (r_1,\ldots,r_k)$}
  \KwOut{A sequence with parentheses and $n+1\ \square$ }
  $w \leftarrow \square$\;
  \If{$F$ is not empty}{
    $(w_1 \cdots w_{k+1}) \leftarrow$ Create a tuple of words \;
  $w_1 \leftarrow \algoName{CFtoW}(\leftsb(r_1))$\;
  \For{$i \in \{2, \ldots, k+1\}$}{
  $w_i \leftarrow \algoName{CFtoW}(\rightsb(r_{i-1}))$\;
  }
  $w \leftarrow \bm{(}w_1 \cdot w_2 \cdots w_{k+1}\bm{)}$\;
  }
  \Return{w}\;
\end{algorithm}
\end{minipage}
\caption{Algorithms that take a \CF\ as input and return the corresponding \ST\ (algorithm on the left) or Parentheses Word (algorithm on the right).
As one can see, both methods are very similar. 
In both cases, two different inputs imply two different outputs. Hence 
both functions are injections. Since there is the same number of objects of the 
same size in all three cases, the functions are bijections.}
\end{figure}

Algorithm~\ref{algo:CF_to_ST} is an injective function: at each level,
the number of subtrees added to the \ST\ is exactly $k+1$,
where $k$ is the number of trees in the \CF.
Since both objects share the same generating function, we obtain the result announced in Lemma~\ref{lm:CF_to_ST}.

\subsection{A bijection with Parentheses Words}

\begin{definition}[Parentheses Word $w$]
\label{def:parentheses}
A \emph{Parentheses Word} $w$ is a word
over the alphabet $\{ \bm{(} ~,~  \square ~,~ \bm{)}  \}$ such that
either $w=\square $ or $w=\bm{(}w_1 \cdots w_k\bm{)}$ 
    where $k\ge 2$ and each $w_i$ is a Parentheses Word.
\end{definition}

Do note that, in this definition, unlike the more commonly found definition of these words, we allow parentheses to surround the entire sequence. We notably do so in order to simplify Algorithm~\ref{algo:CF_to_W} and to make the bijection more apparent to the reader.
But since every parentheses word that is not
$\square$ is contained between a pair of parentheses, it is not necessary to represent it. 
Informally, one may consider the symbols of the word as separators between the nodes and each group of parentheses as a (sub-)forest.
Using the same arguments as for Schröder Tree, Algorithm~\ref{algo:CF_to_W} is an injective function.

\begin{lemma}\label{lm:CF_to_ST}
Algorithms~\ref{algo:CF_to_ST} and~\ref{algo:CF_to_W} are bijective functions that map a \CF\ with $n$ nodes to a \ST\ with $n+1$ leaves (Algorithm~\ref{algo:CF_to_ST}), and to a Parentheses Word with $n+1$ symbols (Algorithm~\ref{algo:CF_to_W}).
They both have a $\Theta(n)$ time-complexity and a $\Theta(n)$ worst-case space complexity.
\end{lemma}

\section{\CF~Signature and \CF~Matching using a filter \label{sec:sign}}
In~\cite{demaine09}, the authors propose a \emph{signature} (or perfect hash) of a
\CT, based on the \SNRepres.
Given a \CT\ with $n$~nodes, its signature is an integer with at most $2n$ bits.
In this section, we extend this notion of signature to \CF s, obtaining
a signature with at most $3n$ bits.

\begin{definition}[\CF ~Signature]
Given  
the \FSkipped ~representation ~$\FSN_x$ of a sequence $x$ (see Definition~\ref{def:FSN}), its signature is defined in the following way: given a position $i \in \{1,\ldots, n\}$
in the \FSkipped ~representation, each value $\FSN_x[i]$ is encoded by a sequence of bits, that are concatenated to obtain an integer.
\begin{itemize}
    \item The first bits concern the sign of $\FSN_x[i]$: If $\FSN_x[i] = 0$, it is equal
    to~${\tt 0}$; The case $\FSN_x[i] < 0$ is encoded by ${\tt 10}$, and the case $\FSN_x[i] > 0$ is encoded by ${\tt 11}$.
    \item If $\FSN_x[i] \neq 0$, the following bits are a unary encoding of $|\FSN_x[i]|$: 
    $|\FSN_x[i]|-1$ bits equal to ${\tt 1}$ followed by a bit equals to ${\tt 0}$.
\end{itemize}
\end{definition}

Since the total number of skipped nodes cannot exceed $n$, a signature contains at most $3n$ bits.
This representation could be used efficiently as a perfect hash for small patterns,
to obtain an efficient algorithm to solve the \CF ~Matching problem: one
only needs to update the signature as one would update the \FSkipped ~representation of the chunk of text
compared to the pattern, using bitwise operations.
But it becomes inefficient when $3n$ exceeds the size of a register, as one cannot
use those operations efficiently anymore.
Therefore, in order to accelerate the \CF ~Matching in a Rabin-Karp fashion we now introduce a filtering method.
\begin{definition}[\CF ~$\tau$-Filter]
Given a sequence $x$ of length $n$ and its linear representation $\LN_x$, 
its \emph{$\tau$-Filter}, denoted by $\Fil_x$, is a sequence of $\tau$ bits such that $\Fil_x[i] = 0$ if $\LN_x[n-\tau+i] = 0$ and $\Fil_x[i] = 1$ otherwise.
\end{definition}

Algorithm~\ref{algo:meta} can be adapted: line $5$ is only tested if 
$\Fil_x$ is equal to $\Fil_p$. $\tau$ can be adapted to match
the size of a classical register, that is $32$, $64$ or $128$ bits.
Therefore, comparing two filters can be made in constant time.
The update of the filter can also be made in constant time by tracking
the positions in $\LN_x$ that have been updated (see~\cite{auvray2025approximatecartesiantreematching} for more details on the update function of the linear representations).
In the following Section, we only implemented the filter for the \FSkipped\  representation.

\section{Experiments \label{sec:XP}}
In this Section, we implemented a \CF ~$\tau$-Filter using
$\tau=64$.
The first set of experiments considers the uniform distribution over text and patterns of respective
length $n$ and $m$, over a $k$ letter alphabet. 
\figurename~\ref{xp_diff_pattern} sums up the details
and observations. 
As one can see, the \FPD\ version
is the slowest ({\tt PD meta}), then comes the one using the \FSkipped ~representation ({\tt SN meta}). The method using a filter ({\tt SN hash}) is even faster, but only with a slight margin.
Finally, the KMP version~\cite{PALP19} is the fastest one ({\tt PD KMP}). Note that we tried to implement the Skip-Search algorithm~\cite{SGRFLP21} without the SIMD instructions, but the experimental
results were disappointing. 
The average cost slightly increases with $k$, which is probably due to the cost of the update function, since the average \FPD ~increases with the size of the alphabet. 
Note that, whatever the linear representation is, the number of comparisons at Line $5$ in Algorithm~\ref{algo:meta} is the same, since they encode the same prefix.
Therefore, it is likely that the extra cost of the \FPD\ representation is due to the fact that both requires to use a stack in their implementations
but more values are pushed in the stack, in the \FPD\ case.

\begin{figure}[ht!]
\begin{minipage}{0.32\textwidth}
\includegraphics[scale=0.25]{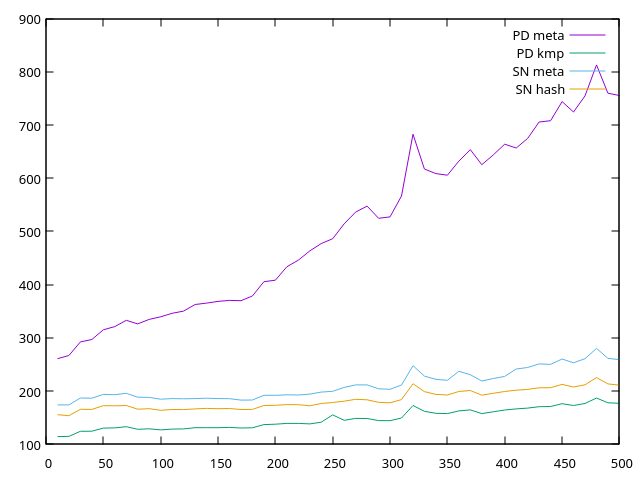}
\end{minipage}
\begin{minipage}{0.32\textwidth}
\includegraphics[scale=0.25]{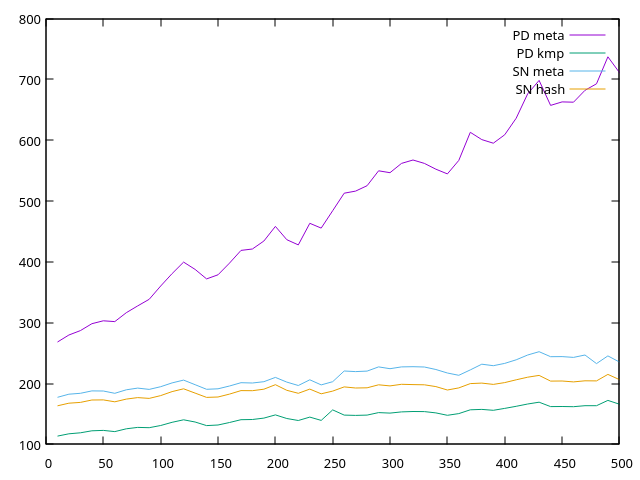}
\end{minipage}
\begin{minipage}{0.33\textwidth}
\includegraphics[scale=0.25]{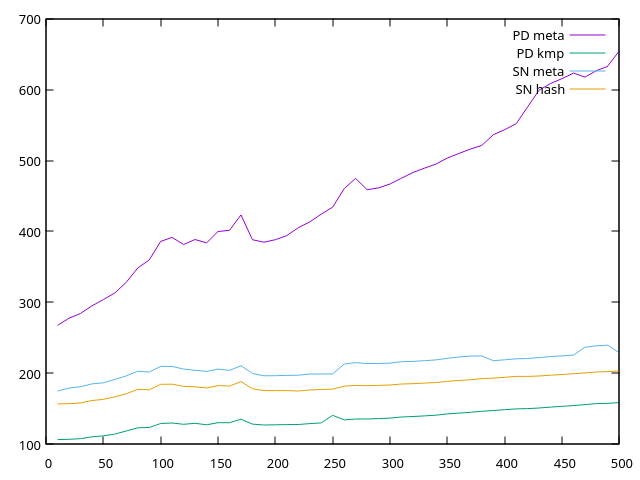}
\end{minipage}
\caption{\label{xp_diff_pattern}%
We randomly generated $10\,000$
patterns of size $m$ and texts of size 
$1\,000$. The $x$-axis represents $m$ and the $y$-axis is the average time taken by each algorithm in microseconds. From left to right, the alpĥabet size is respectively equal to $2$, $4$ and  $m$.
}
\end{figure}
The experiment in \figurename~\ref{xp:ent}  uses the random generator from~\cite{david:hal-04616899}. 
It shows the algorithm remains efficient even with (not too) low entropy.
It generates sequences uniformly amongst those of a fixed length over a fixed alphabet 
and a fixed Collision Entropy (that is Rényi Entropy with $\alpha = 2$).
The Collision Entropy is a function of the probability for two random variables (a letter in the pattern $p$ and the text $t$) to be equal.
The lower the entropy is, the higher average complexity of the algorithms. As a matter of fact, if the entropy reaches $0$, then both
the text and the pattern contains the repetition of a unique symbol, which corresponds to the worst-case complexity of Algorithm~\ref{algo:meta}.
Though, as one can see in \figurename~\ref{xp:ent},
the average cost of the algorithms
quickly drops and stabilizes.
The KMP version is the most efficient (this is in line with the results of~\cite{AmirAFS24} showing that the KMP algorithm is efficient for order preserving matching) and is (unsurprisingly) very stable,
even when the entropy is low.

\begin{figure}[ht!]
\begin{minipage}{0.49\textwidth}
\includegraphics[scale=0.37]{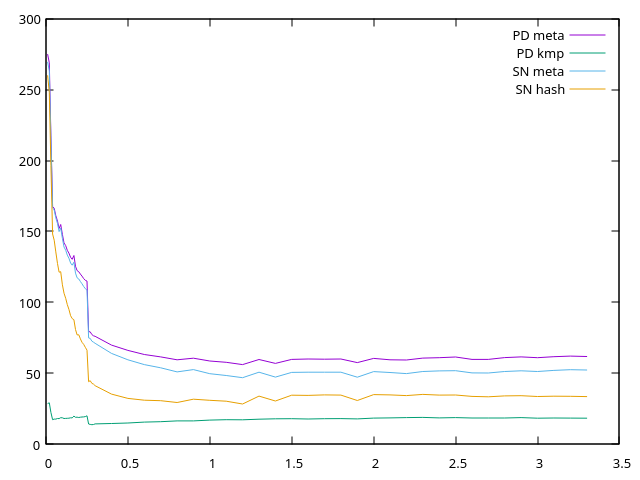}
\end{minipage}
\begin{minipage}{0.49\textwidth}
\includegraphics[scale=0.37]{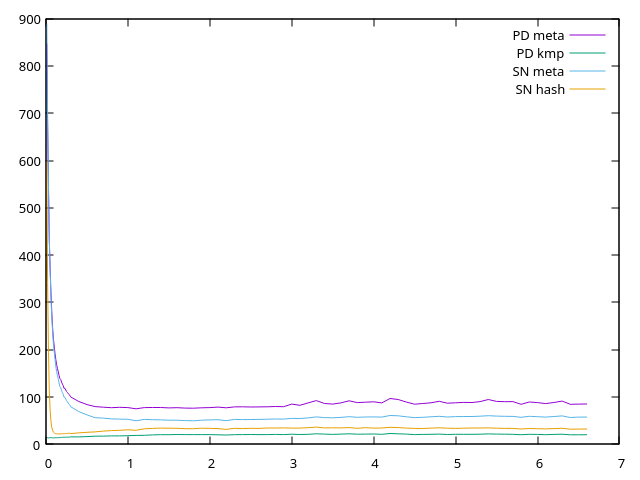}
\end{minipage}

\caption{\label{xp:ent}%
For each value of the collision entropy ($x$-axis), 
$10\,000$ random texts of length $1\,000$ and patterns of length $100$ were generated, 
and the four algorithms were applied to them. 
On the left $m=10$, and on the right $m=100$.
The $y$-axis is the average time in microseconds.
The added value of using the filter method, compared to the simple 
\FSkipped ~version, is more important when the entropy is low, 
whereas the difference in efficiency between the \FPD ~version and the \FSkipped ~version drops.%
}
\end{figure}

\bibliographystyle{splncs04}
\bibliography{biblio}

\newpage
\appendix
\section{Appendix \label{sec:app}}
\subsection{On the Linear Representations of Cartesian Forests}

We show that two sequences $x$ and $y$ share the same \CF\ if and only if they share the same \FPD\ representation.

\begin{proof}
If the sequences have different lengths, then they have different \CF s and different \FPD\ representations and it holds. We need only consider the case when $x$ and $y$ have the same length $m$. We use a proof by induction on $m$.\\
For $m=0$, both \CF s and their representations are empty and trivially match.
For $m=1$, both \CF s have only one root and are the same, and their \FPD\ representation are both equal to $0$. The property holds for $m=1$.
Let us now assume that the property holds for a certain $m=k$ and move onto $m = k+1$.\\
{\bf Same \CF $\implies$ same \FPD }:
Suppose that $x[1\ldots k+1]$ and $y[1\ldots k+1]$ share the same \CF, two cases arise.
\begin{itemize}
    \item If positions $k+1$ have a parent in both \CF s, 
    let $\alpha$ (resp. $\beta$) be the parent of $k+1$ in $F(x[1\ldots k+1])$ (resp. $F(y[1\ldots k+1])$). 
    Since both sequences share the same \CF, we must have $\alpha = \beta$. 
    We have $\FPDD_x[1\ldots k] = \FPDD_y[1\ldots k]$ by induction hypothesis.
    If $\alpha$ and $\beta$ are parents then 
    $\FPDD_x[k+1] = \FPDD_y[k+1] = k+1-\alpha = k+1-\beta$.
    Therefore $\FPDD_x[1\ldots k+1] = \FPDD_y[1\ldots k+1]$.
    
    \item If positions $k+1$ have siblings in both \CF s, 
    let $\alpha$ (resp. $\beta$) be the sibling of $k+1$ in
    $F(x[1\ldots k+1])$ (resp. $F(y[1\ldots k+1])$). 
    The same reasoning as before applies, except
    $\FPDD_x[k+1] = \FPDD_y[k+1] = -(k+1-\alpha) = -(k+1-\beta)$.
    Therefore $\FPDD_x[1\ldots k+1] = \FPDD_y[1\ldots k+1]$.
    
    \item If positions $k+1$ possess neither parents or siblings, then we have $\FPDD_x[k+1] = \FPDD_y[k+1] = 0$ and thus $\FPDD_x[1\ldots k+1] = \FPDD_y[1\ldots k+1]$.
\end{itemize}
{\bf Same \FPD $\implies$ same \CF}:
Suppose that $\FPDD_x[1\ldots k+1] = \FPDD_y[1\ldots k+1]$.
It implies $\FPDD_x[1\ldots k] = \FPDD_y[1\ldots k]$
and also $F(x[1\ldots k]) = F(y[1\ldots k])$ by induction hypothesis. 
From $F(x[1\ldots k])$, we obtain $F(x[1\ldots k+1])$ in the following way. 
\begin{itemize}
    \item 
    If $\FPDD_x[k+1] = \FPDD_y[k+1]\neq 0$, let $\alpha = k+1 - | \FPDD_x[k+1] |$.
    \begin{itemize}
         \item If $\FPDD_x[k+1] < 0$, then $\alpha$ is the sibling of $k+1$ in both $F(x[1\ldots k+1])$ and $F(y[1\ldots k+1])$. 
        \item Otherwise, if $\FPDD_x[k+1] > 0$, that means $\alpha$ is the parent of $k+1$ in $F(x[1\ldots k+1])$ and $F(y[1\ldots k+1])$. Also the right sub-Forest of $\alpha$ becomes the left sub-Forest of $k+1$.
    \end{itemize}
    
    \item
    If $\FPDD_x[k+1] = 0$, $k+1$ is the root of $F(x[1\ldots k+1])$ and $F(x[1\ldots k])$ becomes the left sub-Forest of $x[k+1]$.
    We derive $F(y[1\ldots k+1])$ from $F(y[1\ldots k])$ in the same way.
\end{itemize}
In all three cases, we can conclude that $F(x[1\ldots k+1]) = F(y[1\ldots k+1])$.

Hence, there is a one-to-one mapping between \CF s and \FPD\ representations.
\end{proof}

A similar proof using the same arguments can be used to prove a one-to-one mapping between \CF s and the \FSkipped\ representation.

\subsection{On the Symbolic Method}

The Symbolic Method was introduced in~\cite{flajolet2009analytic}.
It is a method to directly translate a recursive definition
of a combinatorial class into a generation function. 
Recall that a combinatorial class is a set of objects associated
with a size function, such that for all $n\ge 0$, the subset
of objects of size $n$ is finite.
This section only contains the information necessary to 
understand how we obtained the generating function $F(z)$.
Therefore, we focus on unlabelled objects and simple
operations. A reader willing to learn more should read~\cite{flajolet2009analytic}.

Let $\mathcal{C}$ be a combinatorial class and 
$\mathcal{C}_n$ be the set of objects of $\mathcal{C}$ of size $n$. We have  $\mathcal{C} = \bigcup_{n\ge 0} \mathcal{C}_n$
The generating function $C(z)$ of $\mathcal{C}$ is defined as a power series 
$$C(z) = \sum_{n\ge 0} c_n z^n,$$
where $c_n = | \mathcal{C}_n|$.

The symbolic method is a dictionary that translates operations
on sets into operations onto their associated generating functions.

\paragraph{The Null size element.} Assume $\mathcal{C} = \mathcal{C}_0$, that is it only contains an object of size $0$.
Then $C(z)=1$.

\paragraph{The size 1 element.} Assume $\mathcal{C} = \mathcal{C}_1$, that is it only contains an object of size $1$.
In our case, this would be a tree containing a unique node.
Then $C(z)= z$.

\paragraph{Disjoint Union.} Assume $\mathcal{C} = D \cup E$,
where $D$ and $E$ are two combinatorial classes.
Then $C(z) = D(z) + E(z)$.

\paragraph{Cartesian Product} Assume $\mathcal{C} = D \times E$,
where $D$ and $E$ are two combinatorial classes.
Then $C(z) = D(z) \cdot E(z)$.

We recall the recursive decomposition that we obtained for
\CF:
$$
\begin{cases}
\includegraphics[scale=0.2]{figures/grammaire_ligne1-eps-converted-to.pdf}\\
\includegraphics[scale=0.2]{figures/grammaire_ligne2-eps-converted-to.pdf}\\
\end{cases}
$$

The first combinatorial class $F$ is defined as:
\begin{itemize}
\item
either an empty forest (marked with $\emptyset$ whose
associated generating function is $1$).
\item  
or a Cartesian product between a sequence of siblings,
whose generating function will be denoted $S(z)$ and 
a tree, that is decomposed into three parts:
\begin{itemize}
    \item a root, that is a tree of size $1$, whose generating 
    function is $z$
    \item a left and a right sub-forest, that is to say two Cartesian Forests.
\end{itemize}
The generating function of this tree is therefore equal to
$F(z) \cdot z \cdot F(z)$.
\end{itemize}
Using the Symbolic Method we therefore obtain that 
$F(z) = 1+ z\cdot F^2(z)\cdot S(z)$.
Using the same ideas, we obtain that 
$$ S(z) = z\cdot F(z)\cdot S(z) +1$$
$$\implies S(z) - z\cdot F(z)\cdot S(z) =  1$$
$$\implies S(z) =  \frac{1}{1- z\cdot F(z)}$$
And therefore
$$F(z) = 1 + \frac{z\cdot F^2(z)}{1- z\cdot F(z)}$$
$$\implies F(z) \cdot(1- z\cdot F(z)) = 1- z\cdot F(z) + z\cdot F^2(z)$$
$$\implies F(z)- z\cdot F^2(z) = 1- z\cdot F(z) + z\cdot F^2(z)$$
$$\implies  2z\cdot F^2(z) - (1+z)\cdot F(z) + 1 = 0$$
Solving this equation we obtain
$\Delta = b^2 - 4ac = (1+z)^2 - 8z = 1 - 6z + z^2$
and 
$$F(z) = \frac{-b + \sqrt{\Delta}}{2a} = \frac{1+z+\sqrt{1 - 6z + z^2}}{4z} $$
Note here that we do not consider $\frac{-b - \sqrt{\Delta}}{2a}$
because we know $F(z)$ must have positive coefficients.

\subsection{Bijection with Schröder Trees and Parentheses Words}

In  \figurename~\ref{fig:bijection} below, we omit to draw
the parentheses surrounding the entire sequence, contrary to what we did in Section 6.4. This is done in order to match classical representations of the Parentheses Words (we write $\bm{(}\square\square\bm{)}\square$
instead of $\bm{(}\bm{(}\square\square\bm{)}\square\bm{)}$ for instance).

\begin{figure}[ht!]
    \centering
    \includegraphics[scale=1]{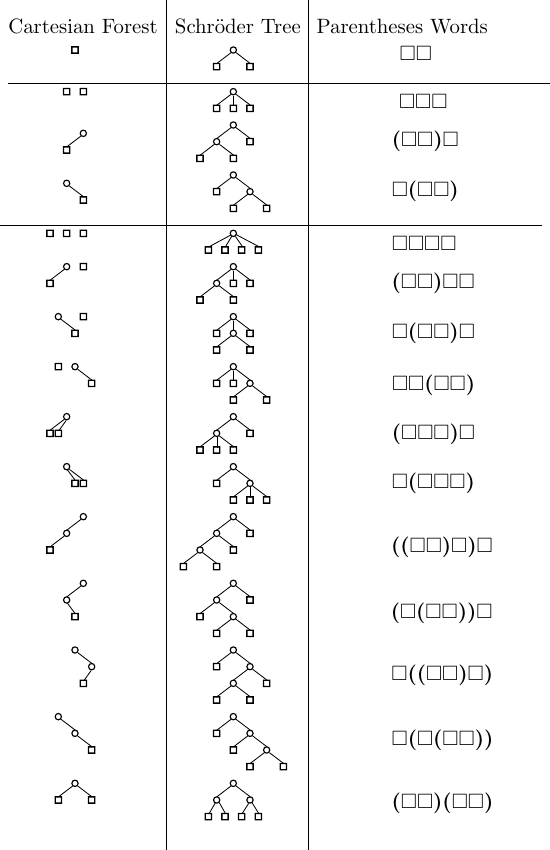}
    \caption{Correspondence between \CF s~with $n$ nodes, \ST s with $n+1$ leaves and sequences of length $n+1$ with parentheses when $n\in\{1,2,3\}$ (respect. top, middle, bottom). $\square$ symbols represent leaves and $\circ$ symbols represent internal nodes.\label{fig:bijection}}
\end{figure}

\end{document}